\documentclass[draft]{article}
\usepackage{amsfonts,amsmath,amsthm,amscd,amssymb,latexsym}

\theoremstyle{plain}
\newtheorem{thm}{Theorem}[section]

\newtheorem{lem}[thm]{Lemma}
\newtheorem{prop}[thm]{Proposition}

\theoremstyle{definition}
\theoremstyle{remark}
\numberwithin{equation}{section}
\newcommand{\keywords}{\textbf{Key words and phrases: }\medskip}
\newcommand{\subjclass}{\textbf{Math. Subj. Clas.: }\medskip}

\begin{document}
\title{\textbf{A discrete version of plane wave solutions of the Dirac equation in the Joyce form} }
\author{\textbf{Volodymyr Sushch} \\
{ Koszalin University of Technology} \\
 { Sniadeckich 2, 75-453 Koszalin, Poland} \\
 {volodymyr.sushch@tu.koszalin.pl} }

\date{}
\maketitle
\begin{abstract}
We construct a  discrete version  of the plane wave solution to a discrete Dirac-K\"{a}hler equation in the Joyce form.
A geometric discretisation scheme  based on both forward and backward difference operators is used. The conditions under which a discrete plane wave solution satisfies a discrete Joyce equation are discussed.
\end{abstract}

\keywords{geometric discretisation, Dirac-K\"{a}hler equation,  Joyce equation, difference equations, discrete models,  plane wave solution}

 \subjclass  {39A12, 39A70, 81Q05}

\section{Introduction}
Discrete models od Dirac type equations based on the Dirac-K\"{a}hler formulation have been of interest recently from both the applied and the theoretical point of view.
In the  Dirac-K\"{a}hler approach, a discretisation scheme is geometric in nature and rests upon the use of the differential forms calculus.
This means that the geometric properties and the algebraic relationships between the differential, the exterior product and the Hodge star operator are expected to be captured in the case of their discrete counterparts.
This work is a continuation of that studied in the papers \cite{S1, S2, S3, S4}. In this paper, we are mainly interested in a discrete model in which both forward and backward differences operators are used. This model relies on the use a geometric discretisation scheme proposed in \cite{S1}, and a discrete Clifford calculus framework on discrete forms described in \cite{S4}.  For a review of discrete Clifford calculus frameworks on lattices, we refer the reader to \cite{Beauce, FKS,F2, F3, Kanamori, Vaz}.  Our purpose here is to construct a discrete version of the plane wave solution to a discrete  Dirac-K\"{a}hler equation in the Joyce form.

We first briefly review some definitions and basic facts  on the Dirac-K\"{a}hler equation \cite{Kahler, Rabin} and the Dirac equation
 in the spacetime algebra \cite{H1, H2}.
Let $M={\mathbb R}^{1,3}$ be  Minkowski space.
Denote by $\Lambda^r(M)$ the vector space of smooth differential $r$-forms, $r=0,1,2,3,4$. We consider  $\Lambda^r(M)$ over $\mathbb{C}$.
Let $\omega,\ \varphi\in\Lambda^r(M)$.
  The inner product is defined by
\begin{equation}\label{1.1}
(\omega, \ \varphi)=\int_{M}\omega\wedge\ast\overline{\varphi},
\end{equation}
where $\wedge$ is the exterior product, $\overline{\varphi}$ denotes the complex conjugate of the form  $\varphi$ and $\ast$ is the Hodge star operator  $\ast:\Lambda^r(M)\rightarrow\Lambda^{4-r}(M)$ with respect to the Lorentz metric.
Let $d:\Lambda^r(M)\rightarrow\Lambda^{r+1}(M)$ be the exterior differential and let $\delta:\Lambda^r(M)\rightarrow\Lambda^{r-1}(M)$ be the formal adjoint of $d$  with respect to  \eqref{1.1}. We have $$\delta=\ast d\ast.$$
 Denote by $\Lambda(M)$ the set of all differential forms on $M$. We have
\begin{equation*}
\Lambda(M)=\Lambda^0(M)\oplus\Lambda^1(M)\oplus\Lambda^2(M)\oplus\Lambda^3(M)\oplus\Lambda^4(M).
\end{equation*}
Let $\Omega\in\Lambda(M)$
be an inhomogeneous differential form, i.e.
\begin{equation}\label{1.2}
\Omega=\sum_{r=0}^4\overset{r}{\omega},
\end{equation}
where $\overset{r}{\omega}\in\Lambda^r(M)$.
 The Dirac-K\"{a}hler equation for a free electron is given by
\begin{equation}\label{1.3}
i(d+\delta)\Omega=m\Omega,
\end{equation}
where $i$ is the usual complex unit  and  $m$  is a mass parameter.
It is easy to show that Eq.~\eqref{1.3} is equivalent to the set of equations
\begin{eqnarray}\label{1.4}
i\delta\overset{1}{\omega}=m\overset{0}{\omega},\nonumber \\
i(d\overset{0}{\omega}+\delta\overset{2}{\omega})=m\overset{1}{\omega},\nonumber \\
i(d\overset{1}{\omega}+\delta\overset{3}{\omega})=m\overset{2}{\omega},\nonumber\\
i(d\overset{2}{\omega}+\delta\overset{4}{\omega})=m\overset{3}{\omega},\nonumber\\
id\overset{3}{\omega}=m\overset{4}{\omega}.
\end{eqnarray}
The operator $d+\delta$ is the analogue of the gradient operator in Minkowski spacetime
\begin{equation*}
\nabla=\sum_{\mu=0}^3\gamma_\mu\partial^\mu, \quad \mu=0,1,2,3,
\end{equation*}
 where $\gamma_\mu$ is the Dirac gamma matrix.
If one think of $\{\gamma_0, \gamma_1, \gamma_2, \gamma_3\}$ as a vector basis in spacetime, then the  gamma matrices $\gamma_\mu$ can be considered as generators of the Clifford algebra of spacetime $\mathcal{C}\ell(1,3)$ \cite{B, B1}.  Hestenes \cite{H2} calls  this algebra the spacetime algebra.   It is known that an inhomogeneous form $\Omega$ can be represented as element of $\mathcal{C}\ell(1,3)$.  Then the Dirac-K\"{a}hler equation can be written as the algebraic equation
 \begin{equation}\label{1.5}
  i\nabla\Omega=m\Omega, \quad \Omega\in\emph{C}\ell(1,3).
 \end{equation}
  Equation~\eqref{1.5} is equivalent to the four Dirac equations (traditional column-spinor equations) for a free electron.  Let $\emph{C}\ell^{ev}(1,3)$ be the even subalgebra of the  algebra $\emph{C}\ell(1,3)$.
Consider the equation
\begin{equation}\label{1.6}
 i\nabla\Omega^{ev}=m\Omega^{ev}\gamma_0 , \quad \Omega^{ev}\in\emph{C}\ell^{ev}(1,3).
 \end{equation}
 In  \cite{J},  this equation  is called  the "generalized bivector Dirac equation". Following Baylis \cite{B1} we call Eq.~\eqref{1.6} the Joyce equation. This equation is equivalent to two copies of the usual Dirac equation. For a deeper discussion of equivalence of Dirac formulations we refer the reader to \cite{JM}.  Equation~\eqref{1.6} admits the plane wave solution of the form
 \begin{equation}\label{1.7}
 \Phi=Ae^{\pm ip\cdot x},
 \end{equation}
 where $A\in\mathcal{C}\ell^{ev}(1,3)$ is a constant element,  $p=\{p_0, p_1, p_2, p_3\}$ is a four-momentum and $p\cdot x=p_\mu x^\mu$ (see \cite{B1} for more details).

It should be noted  that the graded algebra $\Lambda(M)$ endowed with the Clifford multiplication is an example of a Clifford algebra.
 In this case the basis covectors $e^\mu=dx^\mu$  are considered as generators of the Clifford algebra.
Let   $\Lambda^{ev}(M)=\Lambda^0(M)\oplus\Lambda^2(M)\oplus\Lambda^4(M)$.
Denote by $\Omega^{ev}$   the even  part of the form \eqref{1.2}.
 Then Eq.~\eqref{1.6} can be rewritten in terms of inhomogeneous forms as
  \begin{equation}\label{1.8}
 i(d+\delta)\Omega^{ev}=m\Omega^{ev} e^0, \quad \Omega^{ev}\in\Lambda^{ev}(M).
 \end{equation}
 We call this equation the Dirac-K\"{a}hler equation in the Joyce form.

 In this paper, we focus on the construction of a discrete version of the plane wave solution \eqref{1.7} for a discrete counterpart of
 Eq.~\eqref{1.8}. In \cite{S3}, the same problem was considered by using a discretisation scheme based on a combinatorial double complex construction. However, this construction generates difference operators of the forward type only and it is not enough to obtain an exact geometric counterpart of the Clifford algebra $\emph{C}\ell(1,3)$. For this purpose, both forward and backward differences are needed. This fact is well-known \cite{Rabin}. As a new result, we announce here the construction of a discrete plane wave solution based on both forward and backward difference operators. This seems to suggest that such a discrete model will be more suitable for  applied purposes.

\section{Combinatorial model and difference operators}

A combinatorial model of Minkowski space  and a discretisation scheme are adopted from \cite{S2}.
For the convenience of the reader we briefly repeat the relevant material from \cite{S2}
without proofs, thus making our presentation self-contained.
Following \cite{Dezin}, let the tensor product  $C(4)=C\otimes C\otimes C\otimes C$
of a   1-dimensional complex be a combinatorial model of Euclidean space
 ${\mathbb R}^4$. The 1-dimensional complex $C$ is defined in the following way.
 Introduce  the sets  $\{x_\kappa\}$ and $\{e_{\kappa}\}$, $\kappa\in {\mathbb Z}$.
Let $C^0$ and $C^1$ be the free abelian groups of 0-dimensional and 1-dimensional
chains generated by $\{x_\kappa\}$ and $\{e_{\kappa}\}$.
The free abelian group is understood as the direct sum of infinity cyclic groups generated by $\{x_\kappa\}$, $\{e_{\kappa}\}$.
 The boundary operator $\partial: C^0\rightarrow 0,  \quad   \partial: C^1\rightarrow C^0$
 is given by
\begin{equation}\label{2.1}
\partial x_\kappa=0, \qquad  \partial e_\kappa=x_{\kappa+1}-x_\kappa.
\end{equation}
The definition is extended to arbitrary chains by linearity.
The direct sum $C=C^0\oplus C^1$ with the boundary operator $\partial$
defines the 1-dimensional complex. It is known that a free abelian group is an abelian group with basis.   One can regard the sets $\{x_\kappa\}$,  $\{e_{\kappa}\}$ as  sets of basis elements of the groups  $C^0$ and $C^1$. Geometrically we can interpret the 0-dimensional basis elements $x_\kappa$ as points of the real line and the 1-dimensional basis elements $e_\kappa$ as open intervals between points, i.e. $e_\kappa=(x_\kappa, x_{\kappa+1})$. We call the complex  $C$ a combinatorial real line.

Multiplying the basis elements $x_\kappa$, $e_\kappa$ in various way we obtain
basis elements of $C(4)$.
Let $s_k$ be an arbitrary  basis element of $C(4)$. Then we have
\begin{equation}\label{2.2}
s_k=s_{k_0}\otimes s_{k_1}\otimes s_{k_2}\otimes s_{k_3},
\end{equation}
 where $s_{k_\mu}$
is either  $x_{k_\mu}$ or  $e_{k_\mu}$ and $k=(k_0,k_1,k_2,k_3)$,  $k_\mu\in \mathbb{Z}$. Let us denote by $s_k^{(r)}$ the $r$-dimensional basis element of $C(4)$.
The dimension $r$ of  $s_k^{(r)}$ is given by the number of factors $e_{k_\mu}$ that appear in it. The notation $s_k^{(r)}$ means that
the product \eqref{2.2} contains exactly $r$ $1$-dimensional elements $e_{k_\mu}$ and $4-r$
 $0$-dimensional elements  $x_{k_\mu}$,  and  the superscript $(r)$ indicates also a position  of $e_{k_\mu}$ in $s_k^{(r)}$.
For example, the  2-dimensional basis elements
of $C(4)$ can be written as
\begin{align*}
e_k^{01}=e_{k_0}\otimes e_{k_1}\otimes x_{k_2}\otimes x_{k_3},  \qquad
e_k^{12}=x_{k_0}\otimes e_{k_1}\otimes e_{k_2}\otimes x_{k_3},\nonumber \\
e_k^{02}=e_{k_0}\otimes x_{k_1}\otimes e_{k_2}\otimes x_{k_3}, \qquad
e_k^{13}=x_{k_0}\otimes e_{k_1}\otimes x_{k_2}\otimes e_{k_3}, \nonumber \\
e_k^{03}=e_{k_0}\otimes x_{k_1}\otimes x_{k_2}\otimes e_{k_3}, \qquad
e_k^{23}=x_{k_0}\otimes x_{k_1}\otimes e_{k_2}\otimes e_{k_3}.
\end{align*}
Let $C(p)$ is the tensor product of $p$ factors of $C=C(1)$, $p=1,2,3$.
The definition \eqref{2.1} of  $\partial$  is  extended to an arbitrary basis element of $C(4)$ by induction on $p$.
Suppose that the boundary operator has been defined for any basis element $s_{k}\in C(p)$. Then we introduce it for the basis element $s_{k_\mu}\otimes s_{k}\in C(p+1)$ by the rule
\begin{equation}\label{2.3}
\partial(s_{k_\mu}\otimes s_{k})=\partial s_{k_\mu}\otimes s_{k}+Q(k_\mu)s_{k_\mu}\otimes\partial s_{k},
\end{equation}
where $s_{k_\mu}\in C$ and
$Q(k_\mu)$ is equal to $+1$ if $s_{k_\mu}=x_{k_\mu}$ and to $-1$ if $s_{k_\mu}=e_{k_\mu}$.
The operation \eqref{2.3} is linearly extended to arbitrary chains.
It is easy to check that  $\partial\partial a=0$ for any chain $a\in C(4)$.

Let $C(4)$ be a  combinatorial model of Minkowski space.  In what follows we assume that the index $k_0$ in \eqref{2.2} corresponds to the time coordinate of
$M$.  Hence, the basis elements of the indicated 1-dimensional complex $C$ will be written as  $x_{k_0}$ and $e_{k_0}$.

Let us now consider a dual complex to $C(4)$. We define it as the complex of cochains
$K(4)$ with complex coefficients. The complex $K(4)$
has the same  structure as $C(4)$, namely ${K(4)=K\otimes K\otimes K\otimes K}$, where $K$ is a dual
complex to the 1-dimensional complex $C$.  We will use a superscript to indicate a basis element of $K$. Let $x^\kappa$ and $e^\kappa$, $\kappa\in {\mathbb Z}$, be  the 0- and 1-dimensional basis elements of $K$. Then an arbitrary $r$-dimensional basis element of $K(4)$ can be written  as
$s_{(r)}^k=s^{k_0}\otimes s^{k_1}\otimes s^{k_2}\otimes s^{k_3}$, where $s^{k_{\mu}}$
is either  $x^{k_{\mu}}$ or  $e^{k_{\mu}}$ and $k=(k_0,k_1,k_2,k_3)$.   We will call cochains forms,
emphasizing their relationship with differential forms.
Denote by  $K^r(4)$ the set of all $r$-forms. Then $K(4)$ can be expressed by
\begin{equation*}
K(4)=K^0(4)\oplus K^1(4)\oplus K^2(4)\oplus K^3(4)\oplus K^4(4).
\end{equation*}
   The complex $K(4)$ is a discrete analogue of $\Lambda(M)$.
  Let $\overset{r}{\omega}\in K^r(4)$, then we have
  \begin{equation}\label{2.4}
  \overset{r}{\omega}=\sum_k\sum_{(r)} \omega_k^{(r)}s_{(r)}^k,
\end{equation}
where  $\omega_k^{(r)}\in\mathbb{C}$.

As in \cite{Dezin}, we define the pairing (chain-cochain) operation for any basis elements
$\varepsilon_k\in C(4)$,  $s^k\in K(4)$ by the rule
\begin{equation}\label{2.5}
\langle\varepsilon_k, \ s^k\rangle=\left\{\begin{array}{l}0, \quad \varepsilon_k\ne s_k\\
                            1, \quad \varepsilon_k=s_k.
                            \end{array}\right.
\end{equation}
The operation \eqref{2.5} is linearly extended to arbitrary chains and cochains.

The coboundary operator $d^c: K^r(4)\rightarrow K^{r+1}(4)$ is defined by
\begin{equation}\label{2.6}
\langle\partial a, \ \overset{r}{\omega}\rangle=\langle a, \ d^c\overset{r}{\omega}\rangle,
\end{equation}
where $a\in C(4)$ is an $r+1$ dimensional chain. The operator $d^c$ is an analog of the exterior differential.
From the above it follows that
\begin{equation*}\label{}
 d^c\overset{4}{\omega}=0 \quad \mbox{and} \quad d^cd^c\overset{r}{\omega}=0 \quad \mbox{for any} \quad r.
\end{equation*}
It is convenient to
introduce the shift operators  $\tau_\mu$ and $\sigma_\mu$ in the set of indices by
\begin{equation}\label{2.7}\tau_\mu k=(k_0,...
 k_\mu+1,...k_3), \quad
 \sigma_\mu k=(k_0,...k_\mu-1,...k_3), \quad \mu=0,1,2,3.
  \end{equation}
    Let the difference operators $\Delta^+_\mu$ and $\Delta^-_\mu$ be defined by
\begin{equation}\label{2.8}
\Delta^+_\mu\omega_k^{(r)}=\omega_{\tau_\mu k}^{(r)}-\omega_k^{(r)},
\end{equation}
\begin{equation}\label{2.9}
\Delta^-_\mu\omega_k^{(r)}=\omega_{k}^{(r)}-\omega_{\sigma_{\mu} k}^{(r)},
\end{equation}
where  $\omega_k^{(r)}\in\mathbb{C}$ is a component of $\overset{r}{\omega}\in K^r(4)$.
It is clear, that
\begin{equation*}\label{}
\Delta^-_\mu\omega_k^{(r)}=\Delta^+_\mu\omega_{\sigma_{\mu} k}^{(r)}, \qquad
\Delta^+_\mu\omega_k^{(r)}=\Delta^-_\mu\omega_{\tau_{\mu} k}^{(r)}.
\end{equation*}
Note that it is enough to use one of the two above defined difference operators to describe discrete analogs of $d$ and $\delta$, as it is shown in \cite{S1},  but in the proposed approach both forward and backward differences are needed to construct a discrete version of the plane wave solution.
It is helpful here to write \eqref{2.4} more explicitly as
\begin{equation*}
\overset{0}{\omega}=\sum_k\overset{0}{\omega}_kx^k,  \quad  \overset{2}{\omega}=\sum_k\sum_{\mu<\nu} \omega_k^{\mu\nu}e_{\mu\nu}^k, \quad \overset{4}{\omega}=\sum_k\overset{4}{\omega}_ke^k,
\end{equation*}
\begin{equation*}
\overset{1}{\omega}=\sum_k\sum_{\mu=0}^3\omega_k^\mu e_\mu^k, \quad
\overset{3}{\omega}=\sum_k\sum_{\iota<\mu<\nu} \omega_k^{\iota\mu\nu}e_{\iota\mu\nu}^k,
\end{equation*}
where
\begin{equation}\label{2.10}
x^k=x^{k_0}\otimes x^{k_1}\otimes x^{k_2}\otimes x^{k_3}, \quad  e^k=e^{k_0}\otimes e^{k_1}\otimes e^{k_2}\otimes e^{k_3}
\end{equation}
are the 0-, 4-dimensional basis elements of $K(4)$, and
 $e_\mu^k$, $e_{\mu\nu}^k$ and  $e_{\iota\mu\nu}^k$ are the 1-, 2- and 3-dimensional basis elements of $K(4)$.

Using \eqref{2.3}, \eqref{2.6} and \eqref{2.8} we can calculate
\begin{equation}\label{2.11}
d^c\overset{0}{\omega}=\sum_k\sum_{\mu=0}^3(\Delta^+_\mu\overset{0}{\omega}_k)e_\mu^k, \qquad
d^c\overset{1}{\omega}=\sum_k\sum_{\mu<\nu}(\Delta^+_\mu\omega_k^\nu-\Delta^+_\nu\omega_k^\mu)e_{\mu\nu}^k,
\end{equation}
\begin{align}\label{2.12}
d^c\overset{2}{\omega}=\sum_k\big[(\Delta^+_0\omega_k^{12}-\Delta^+_1\omega_k^{02}+\Delta^+_2\omega_k^{01})e_{012}^k\nonumber \\
+(\Delta^+_0\omega_k^{13}-\Delta^+_1\omega_k^{03}+\Delta^+_3\omega_k^{01})e_{013}^k \nonumber \\
+(\Delta^+_0\omega_k^{23}-\Delta^+_2\omega_k^{03}+\Delta^+_3\omega_k^{02})e_{023}^k \nonumber \\
+(\Delta^+_1\omega_k^{23}-\Delta^+_2\omega_k^{13}+\Delta^+_3\omega_k^{12})e_{123}^k\big],
\end{align}
\begin{equation}\label{2.13}
d^c\overset{3}{\omega}=\sum_k(\Delta^+_0\omega_k^{123}-\Delta^+_1\omega_k^{023}+\Delta^+_2\omega_k^{013}-\Delta^+_3\omega_k^{012})e^k.
\end{equation}

Let us now introduce  a $\cup$-multiplication of discrete forms which is an analog of the
exterior multiplication for differential forms.
Denote by  $K(p)$  the tensor product of $p$ factors of the 1-dimensional complex $K=K(1)$.
For the basis elements of $K$ the $\cup$-multiplication is defined as follows
\begin{equation}\label{2.14}
x^{k_\mu}\cup x^{k_\mu}=x^{k_\mu}, \quad e^{k_\mu}\cup x^{k_\mu+1}=e^{k_\mu},
\quad x^{k_\mu}\cup e^{k_\mu}=e^{k_\mu}, \quad {k_\mu}\in{\mathbb Z},
\end{equation}
supposing the product to be zero in all other case. To  arbitrary basis elements of $K(p)$, $p=2,3,4$, the definition \eqref{2.14}
is extended by induction on  $p$ (see \cite{Dezin} and \cite{S1} for more details).
 Again, to arbitrary discrete forms the
$\cup$-multiplication is extended linearly.

It is important to note that the definition above is suitable to deal with a discrete version of the Leibniz rule. The following result was proven in \cite{Dezin}.

\begin{prop}
Let $\varphi\in K^r(4)$ and $\psi\in K^q(4)$ be arbitrary discrete forms.
Then
\begin{equation}\label{2.15}
 d^c(\varphi\cup\psi)=d^c\varphi\cup\psi+(-1)^r\varphi\cup
d^c\psi.
\end{equation}
\end{prop}

Consider  a discrete analog of the Hodge star operator
$\ast$. For simplicity of  notation, we continue to write $\ast$ in the discrete case.
Define the operation $\ast: K^r(4)\rightarrow K^{4-r}(4)$ for an arbitrary basis element $s^k=s^{k_0}\otimes  s^{k_1}\otimes s^{k_2}\otimes s^{k_3}$ by the rule
\begin{equation}\label{2.16}
 s^k\cup\ast s^k=Q(k_0) e^{k},
\end{equation}
where $Q(k_0)$ is equal to $+1$ if $s^{k_0}=x^{k_0}$ and
to $-1$ if  $s^{k_0}= e^{k_0}$, and $e^k$ is given by  \eqref{2.10}.
 For example, for the 1-dimensional basis elements
 $e_\mu^k$  we have $e_0^k\cup\ast e_0^k=-e^k$ and  $e_\mu^k\cup\ast e_\mu^k=e^k$ \ for $\mu=1,2,3$. It is clear that  the definition  \eqref{2.16}  preserves the  Lorentz  signature of metric in our  discrete model. From  \eqref{2.16}  a more detailed calculation leads to
\begin{align}\label{2.17}
\ast x^k=e^k, \qquad
\ast e^k=-x^{\tau k},
\end{align}
\begin{equation}\label{2.18}
\ast e_0^k=-e_{123}^{\tau_0 k}, \qquad \ast e_1^k=-e_{023}^{\tau_1 k}, \qquad
\ast e_2^k=e_{013}^{\tau_2 k}, \qquad \ast e_3^k=-e_{012}^{\tau_3 k},
\end{equation}
\begin{align}\label{2.19}
\ast e_{01}^k&=-e_{23}^{\tau_{01} k}, \qquad \ast e_{02}^k=e_{13}^{\tau_{02} k}, \qquad \ast e_{03}^k=-e_{12}^{\tau_{03} k}, \nonumber \\
\ast e_{12}^k&=e_{03}^{\tau_{12} k}, \qquad  \ast e_{13}^k=-e_{02}^{\tau_{13} k}, \qquad \ast e_{23}^k=e_{01}^{\tau_{23} k},
\end{align}
\begin{equation}\label{2.20}
\ast e_{012}^k=-e_3^{\tau_{012} k}, \quad \ast e_{013}^k=e_2^{\tau_{013} k}, \quad
\ast e_{023}^k=-e_1^{\tau_{023} k}, \quad \ast e_{123}^k=-e_0^{\tau_{123} k}.
\end{equation}
Here  $\tau_{\mu\nu}k=\tau_{\mu}\tau_{\nu}k$ and $\tau_{\mu\nu\iota}k=\tau_{\mu}\tau_{\nu}\tau_{\iota}k$, where $\tau_{\mu}$ is defined by  \eqref{2.7}, and
\begin{equation}\label{2.21}
\tau k=(k_0+1,\ k_1+1,\ k_2+1,\ k_3+1).
   \end{equation}
   The operation $\ast$ is linearly extended to arbitrary forms.
  It is easy to check that
 \begin{equation*}
\ast\ast s^k_{(r)}=(-1)^{r+1}s^{\tau k}_{(r)},
 \end{equation*}
 where $s^k_{(r)}$ is an $r$-dimensional basic element of $K(4)$.
 Then if we perform the  operation $\ast$ twice on any $r$-form $\overset{r}{\omega}\in K(4)$,  we  obtain
 \begin{equation*}
\ast\ast\overset{r}{\omega}=(-1)^{r+1}\sum_k\sum_{(r)}\omega_k^{(r)}s^{\tau k}_{(r)}=
(-1)^{r+1}\sum_k\sum_{(r)}\omega_{\sigma k}^{(r)}s^k_{(r)},
 \end{equation*}
where $\tau k$ is given by \eqref{2.21} and
\begin{equation}\label{2.22}
\sigma k=(k_0-1,\ k_1-1,\ k_2-1,\ k_3-1).
 \end{equation}
Hence the operation $(\ast)^2$ is equivalent to a shift with corresponding sign. This is slightly different from the continuum case, where applying the Hodge star operator twice leaves a differential form unchanged up to sign,
 i.e. $(\ast)^2=\pm 1$.

Let us consider the  $4$-dimensional finite chain  $e_n\subset C(4)$ of the form:
\begin{equation}\label{2.23}
e_n=\sum_ke_k, \quad k_\mu=1,2, ...,n_\mu,
\end{equation}
where $n_\mu\in \mathbb{N}$ is a fixed number for each $\mu=0,1,2,3$ and $e_k$ is given by \eqref{2.10}.
 The finite sum \eqref{2.23} is the set defined by a finite number of $4$-dimensional basis elements of $C(4)$. This set imitates a domain of $M$.

Suppose that the  $r$-form  \eqref{2.4} is vanished on $C(4)\setminus e_n$, i.e. if $k_\mu<1$ or $k_\mu>n_\mu$ then  $\omega_k^{(r)}=0$ for any $r$ and $k$.
 Then for forms
$\varphi, \ \omega\in K^r(4)$ of the same degree $r$ the inner
 product over the set  $e_n$ \eqref{2.23} is defined  by the rule
 \begin{equation}\label{2.24}
 (\varphi, \ \omega)_{e_n}=\langle e_n, \ \varphi\cup\ast\overline{\omega}\rangle,
 \end{equation}
  where $\overline{\omega}$ denotes the complex conjugate of the form $\omega$. For  forms of different degrees the product \eqref{2.24} is set equal to zero.
 See also \cite{S1}.
The definition \eqref{2.24} imitates correctly the continuum case \eqref{1.1} and the Lorentz metric structure is still captured here. Using \eqref{2.5}, \eqref{2.14} and \eqref{2.17}--\eqref{2.20}
 we obtain
\begin{align*}\label{}
(\overset{0}{\varphi}, \ \overset{0}{\omega})_{e_n}&=\sum_k\overset{0}{\varphi}_k\overset{0}{\overline{\omega}}_k, \qquad
(\overset{4}{\varphi}, \ \overset{4}{\omega})_{e_n}=-\sum_k\overset{4}{\varphi}_k\overset{4}{\overline{\omega}}_k,\\
(\overset{1}{\varphi}, \ \overset{1}{\omega})_{e_n}&=\sum_k\big(-\varphi_k^0\overline{\omega}_k^0+\varphi_k^1\overline{\omega}_k^1+\varphi_k^2\overline{\omega}_k^2+\varphi_k^3\overline{\omega}_k^3\big), \\
(\overset{2}{\varphi}, \ \overset{2}{\omega})_{e_n}&=\sum_k\big(-\varphi_k^{01}\overline{\omega}_k^{01}-\varphi_k^{02}\overline{\omega}_k^{02}-\varphi_k^{03}\overline{\omega}_k^{03}
+\varphi_k^{12}\overline{\omega}_k^{12}+\varphi_k^{13}\overline{\omega}_k^{13}+\varphi_k^{23}\overline{\omega}_k^{23}\big), \\
(\overset{3}{\varphi}, \ \overset{3}{\omega})_{e_n}&=\sum_k\big(-\varphi_k^{012}\overline{\omega}_k^{012}-\varphi_k^{013}\overline{\omega}_k^{013}-\varphi_k^{023}\overline{\omega}_k^{023}
+\varphi_k^{123}\overline{\omega}_k^{123}\big).
\end{align*}
The following result is taken from \cite{S1}.
\begin{prop}
 Let $\overset{r}\varphi\in K^r(4)$  and $\overset{r+1}\omega\in K^{r+1}(4)$,  $r=0,1,2,3$. Then we have
\begin{equation}\label{2.25}
 (d^c\overset{r}\varphi, \ \overset{r+1}\omega)_{e_n}=(\overset{r}\varphi, \ \delta^c\overset{r+1}\omega)_{e_n},
\end{equation}
 where
 \begin{equation}\label{2.26}
 \delta^c\overset{r+1}\omega=(-1)^{r+1}\ast^{-1}d^c\ast\overset{r+1}\omega
 \end{equation}
  is the
operator formally adjoint of $d^c$.
\end{prop}
Here $\ast^{-1}$ is the inverse of $\ast$, i.e. $\ast\ast^{-1}=1$.

The operator $\delta^c: K^{r+1}(4) \rightarrow K^r(4)$ is a discrete analog of the codifferential $\delta$. For the 0-form $\overset{0}{\omega}\in K^0(4)$ we have $\delta^c\overset{0}{\omega}=0$.
It is obvious from \eqref{2.25} that
$\delta^c\delta^c\overset{r}{\omega}=0$ \ for any $r=1,2,3,4.$
 Using \eqref{2.11}--\eqref{2.13} and \eqref{2.25} we can calculate
\begin{equation}\label{2.27}
\delta^c\overset{1}{\omega}=\sum_k(\Delta^-_0\omega_{k}^{0}-\Delta^-_1\omega_{k}^{1}-\Delta^-_2\omega_{k}^{2}-\Delta^-_3\omega_{k}^{3})x^k,
\end{equation}
\begin{align}\label{2.28} \nonumber
\delta^c\overset{2}{\omega}=\sum_k\big[(\Delta^-_1\omega_{k}^{01}+\Delta^-_2\omega_{k}^{02}+\Delta^-_3\omega_{k}^{03})e_{0}^k\\ \nonumber
+(\Delta^-_0\omega_{k}^{01}+\Delta^-_2\omega_{k}^{12}+\Delta^-_3\omega_{k}^{13})e_{1}^k\\ \nonumber
+(\Delta^-_0\omega_{k}^{02}-\Delta^-_1\omega_{k}^{12}+\Delta^-_3\omega_{k}^{23})e_{2}^k\\
+(\Delta^-_0\omega_{k}^{03}-\Delta^-_1\omega_{k}^{13}-\Delta^-_2\omega_{k}^{23})e_{3}^k\big],
\end{align}
\begin{align}\label{2.29} \nonumber
\delta^c\overset{3}{\omega}=\sum_k\big[(-\Delta^-_2\omega_{k}^{012}-\Delta^-_3\omega_{k}^{013})e_{01}^k+
(\Delta^-_1\omega_{k}^{012}-\Delta^-_3\omega_{k}^{023})e_{02}^k\\ \nonumber
+(\Delta^-_1\omega_{k}^{013}+\Delta^-_2\omega_{k}^{023})e_{03}^k
+(\Delta^-_0\omega_{k}^{012}-\Delta^-_3\omega_{k}^{123})e_{12}^k\\
+(\Delta^-_0\omega_{k}^{013}+\Delta^-_2\omega_{k}^{123})e_{13}^k
+(\Delta^-_0\omega_{k}^{023}-\Delta^-_1\omega_{k}^{123})e_{23}^k\big],
\end{align}
\begin{align}\label{2.30}
\delta^c\overset{4}{\omega}=\sum_k\big[(\Delta^-_3\omega_{k}^{4})e_{012}^k-(\Delta^-_2\omega_{k}^{4})e_{013}^k
+(\Delta^-_1\omega_{k}^{4})e_{023}^k+(\Delta^-_0\omega_{k}^{4})e_{123}^k\big].
\end{align}
 Note that formulas~\eqref{2.27}--\eqref{2.30} are essentially the same as the corresponding formulas from \cite{S1} up to notation.

As in the continuum case, the discrete codifferencial does not depend on the signature of the inner product, i.e. the operator $\delta^c$ is the same in both $(+---)$ and $(-+++)$ cases.
Recall that in the continuum case we have  $\delta=\ast d\ast.$ However in the discrete case  the operator $\delta^c$ is not equal to $\ast d^c\ast$. Indeed, let
\begin{equation*}
\delta^c\overset{r+1}{\omega}=\overset{r}{\psi}=\sum_k\sum_{(r)}\psi_k^{(r)} s^k_{(r)}, \quad r=0,1,2,3.
 \end{equation*}
 But an easy computation shows that
 \begin{equation*}
 \ast d^c\ast\overset{r+1}{\omega}=\sum_k\sum_{(r)}\psi_{\sigma k}^{(r)}s^k_{(r)},
 \end{equation*}
 where $\sigma k$ is given by \eqref{2.22}.
 Therefore in comparison with $\delta^c\overset{r+1}{\omega}$, the form $\ast d^c\ast\overset{r+1}{\omega}$ has the same components which are shifting to the left by one by all indexes.
 For example,
 \begin{equation*}
\ast d^c\ast\overset{1}{\omega}=\sum_k(\Delta^-_0\omega_{\sigma k}^{0}-\Delta^-_1\omega_{\sigma k}^{1}-\Delta^-_2\omega_{\sigma k}^{2}-\Delta^-_3\omega_{\sigma k}^{3})x^k
\end{equation*}
while $\delta^c\overset{1}{\omega}$ has  the view \eqref{2.27}.

 \section{Discrete Dirac-K\"{a}hler and Joyce equations}
 The linear map
\begin{equation*}
\Delta^c=-(d^c\delta^c+\delta^cd^c): \ K^r(4) \rightarrow K^r(4)
\end{equation*}
is called a discrete analogue of the Laplacian. It is clear that
\begin{equation*}
-(d^c\delta^c+\delta^cd^c)=(d^c-\delta^c)^2=-(d^c+\delta^c)^2.
\end{equation*}
Then a square root of the discrete Laplacian can be written either as $d^c-\delta^c$ or as $i(d^c+\delta^c)$, where $i$ is the usual complex unit.
Here we consider a discrete version of the Dirac-K\"{a}hler  equation with  $i(d^c+\delta^c)$.
 Note that in \cite{S1},  a discrete model of the Dirac-K\"{a}hler  equation has been constructed by using the operator $d^c-\delta^c$.
Let $\Omega$ be  a discrete inhomogeneous form, that is
\begin{equation}\label{3.1}
\Omega=\sum_{r=0}^4\overset{r}{\omega},
\end{equation}
where $\overset{r}{\omega}$ is given by \eqref{2.4}. An alternative notation includes $(\Omega)_r=\overset{r}{\omega}$ for the $r$-form part of an inhomogeneous form.

  Introduce a discrete analog of the Dirac-K\"{a}hler equation \eqref{1.3} by the rule
\begin{equation}\label{3.2}
i(d^c+\delta^c)\Omega=m\Omega,
\end{equation}
where  $m$ is a positive number (mass parameter).
We can write this equation more explicitly by separating its homogeneous components as
\begin{eqnarray}\label{3.3} \nonumber
i\delta^c\overset{1}{\omega}=m\overset{0}{\omega}, \quad i(d^c\overset{1}{\omega}+\delta^c\overset{3}{\omega})=m\overset{2}{\omega}, \quad
id^c\overset{3}{\omega}=m\overset{4}{\omega},\\
i(d^c\overset{0}{\omega}+\delta^c\overset{2}{\omega})=m\overset{1}{\omega}, \qquad
i(d^c\overset{2}{\omega}+\delta^c\overset{4}{\omega})=m\overset{3}{\omega}.
\end{eqnarray}
Using \eqref{2.11}--\eqref{2.13} and \eqref{2.27}--\eqref{2.30} Eqs.~\eqref{3.3} can be written as the set of following difference equations
\begin{align*}\label{}
i(\Delta^-_0\omega_k^{0}-\Delta^-_1\omega_k^{1}-\Delta^-_2\omega_k^{2}-\Delta^-_3\omega_k^{3})=m\overset{0}{\omega}_k,\\
i(\Delta^+_0\overset{0}{\omega}_k+\Delta^-_1\omega_k^{01}+\Delta^-_2\omega_k^{02}+\Delta^-_3\omega_k^{03})=m\omega_k^0,\\
i(\Delta^+_1\overset{0}{\omega}_k+\Delta^-_0\omega_k^{01}+\Delta^-_2\omega_k^{12}+\Delta^-_3\omega_k^{13})=m\omega_k^1,\\
i(\Delta^+_2\overset{0}{\omega}_k+\Delta^-_0\omega_k^{02}-\Delta^-_1\omega_k^{12}+\Delta^-_3\omega_k^{23})=m\omega_k^2,\\
i(\Delta^+_3\overset{0}{\omega}_k+\Delta^-_0\omega_k^{03}-\Delta^-_1\omega_k^{13}-\Delta^-_2\omega_k^{23})=m\omega_k^3,\\
i(\Delta^+_0\omega_k^1-\Delta^+_1\omega_k^0-\Delta^-_2\omega_k^{012}-\Delta^-_3\omega_k^{013})=m\omega_k^{01},\\
i(\Delta^+_0\omega_k^2-\Delta^+_2\omega_k^0+\Delta^-_1\omega_k^{012}-\Delta^-_3\omega_k^{023})=m\omega_k^{02},\\
i(\Delta^+_0\omega_k^3-\Delta^+_3\omega_k^0+\Delta^-_1\omega_k^{013}+\Delta^-_2\omega_k^{023})=m\omega_k^{03},\\
i(\Delta^+_1\omega_k^2-\Delta^+_2\omega_k^1+\Delta^-_0\omega_k^{012}-\Delta^-_3\omega_k^{123})=m\omega_k^{12},\\
i(\Delta^+_1\omega_k^3-\Delta^+_3\omega_k^1+\Delta^-_0\omega_k^{013}+\Delta^-_2\omega_k^{123})=m\omega_k^{13},\\
i(\Delta^+_2\omega_k^3-\Delta^+_3\omega_k^2+\Delta^-_0\omega_k^{023}-\Delta^-_1\omega_k^{123})=m\omega_k^{23}, \\
i(\Delta^+_0\omega_k^{12}-\Delta^+_1\omega_k^{02}+\Delta^+_2\omega_k^{01}+\Delta^-_3\overset{4}{\omega}_k)=m\omega_k^{012},\\
i(\Delta^+_0\omega_k^{13}-\Delta^+_1\omega_k^{03}+\Delta^+_3\omega_k^{01}-\Delta^-_2\overset{4}{\omega}_k)=m\omega_k^{013},\\
i(\Delta^+_0\omega_k^{23}-\Delta^+_2\omega_k^{03}+\Delta^+_3\omega_k^{02}+\Delta^-_1\overset{4}{\omega}_k)=m\omega_k^{023},\\
i(\Delta^+_1\omega_k^{23}-\Delta^+_2\omega_k^{13}+\Delta^+_3\omega_k^{12}+\Delta^-_0\overset{4}{\omega}_k)=m\omega_k^{123}, \\
i(\Delta^+_0\omega_k^{123}-\Delta^+_1\omega_k^{023}+\Delta^+_2\omega_k^{013}-\Delta^+_3\omega_k^{012})=m\overset{4}{\omega}_k.
\end{align*}
 Note that the equations above contain the difference operators  both the forward and backward type. This is in contrast with the situation    described in \cite{S1}, where difference operators of the forward type only are used for the discrete construction.

Let us define  the Clifford multiplication in  $K(4)$ by the following rules:
\begin{align*}
&\mbox{(a)} \quad x^kx^k=x^k, \quad x^ke^k_\mu=e^k_\mu x^k=e^k_\mu,\\
&\mbox{(b)} \quad e^k_\mu e^k_\nu+e^k_\nu e^k_\mu=2g_{\mu\nu}x^k, \quad g_{\mu\nu}=\mbox{diag}(1,-1,-1,-1),\\
&\mbox{(c)} \quad e^k_{\mu_1}\cdots e^k_{\mu_s}=e^k_{\mu_1\cdots \mu_s} \quad \mbox{for} \quad 0\leq \mu_1<\cdots <\mu_s\leq 3,
\end{align*}
supposing the product to be zero in all other cases.

The operation is linearly extended to arbitrary discrete forms.

Consider the following unit forms
\begin{equation}\label{3.4}
x=\sum_kx^k, \quad e=\sum_ke^k, \quad e_\mu=\sum_ke_\mu^k, \quad e_{\mu\nu}=\sum_ke_{\mu\nu}^k.
\end{equation}
Note that the unit 0-form $x$ plays  a role of the unit element in $K(4)$ with respect to the Clifford multiplication, i.e. for any $r$-form  $\overset{r}{\omega}$ we have
\begin{equation*}
x\overset{r}{\omega}=\overset{r}{\omega}x=\overset{r}{\omega}.
\end{equation*}
The following is straightforward.
\begin{prop}
\begin{equation}\label{3.5}
e_\mu e_\nu+e_\nu e_\mu=2g_{\mu\nu}x, \qquad \mu,\nu=0,1,2,3.
\end{equation}
\end{prop}
\begin{prop}  For any inhomogeneous form $\Omega\in K(4)$ we have
\begin{equation}\label{3.6}
(d^c+\delta^c)\Omega=\sum_{r=0}^3\Big(\sum_{\mu=0}^3e_\mu\Delta^+_\mu\overset{r}{\omega}\Big)_{r+1}+
\sum_{r=1}^4\Big(\sum_{\mu=0}^3e_\mu\Delta^-_\mu\overset{r}{\omega}\Big)_{r-1},
\end{equation}
where
 $\Delta^+_\mu$  and $\Delta^-_\mu$ are the difference operators which act on each component of $\overset{r}{\omega}$ by the rules \eqref{2.8} and \eqref{2.9}.
\end{prop}
\begin{proof}
By \eqref{2.11} it is clear that
\begin{align*}
d^c\overset{0}{\omega}=\sum_{\mu=0}^3e_\mu\Delta^+_\mu\overset{0}{\omega}=\Big(\sum_{\mu=0}^3e_\mu\Delta^+_\mu\overset{0}{\omega}\Big)_1.
\end{align*}
In the case of the 2-form $\overset{2}{\omega}$ we compute
\begin{align*}
\sum_{\mu=0}^3e_\mu\Delta^{\pm}_\mu\overset{2}{\omega}=\Big(\sum_{\mu=0}^3e_\mu\Delta^{\pm}_\mu\overset{2}{\omega}\Big)_{1}+
\Big(\sum_{\mu=0}^3e_\mu\Delta^{\pm}_\mu\overset{2}{\omega}\Big)_{3}\\=
\sum_k[(\Delta^{\pm}_1\omega_k^{01}+\Delta^{\pm}_2\omega_k^{02}+\Delta^{\pm}_3\omega_k^{03})e_0^k\\
+(\Delta^{\pm}_0\omega_k^{01}+\Delta^{\pm}_2\omega_k^{12}+\Delta^{\pm}_3\omega_k^{13})e_1^k\\
+(\Delta^{\pm}_0\omega_k^{02}-\Delta^{\pm}_1\omega_k^{12}+\Delta^{\pm}_3\omega_k^{23})e_2^k\\+
(\Delta^{\pm}_0\omega_k^{03}-\Delta^{\pm}_1\omega_k^{13}-\Delta^{\pm}_2\omega_k^{23})e_3^k]\\
+\sum_k[(\Delta^{\pm}_0\omega_k^{12}-\Delta^{\pm}_1\omega_k^{02}+\Delta^{\pm}_2\omega_k^{01})e_{012}^k\\+
(\Delta^{\pm}_0\omega_k^{13}-\Delta^{\pm}_1\omega_k^{03}+\Delta^{\pm}_3\omega_k^{01})e_{013}^k\\
+(\Delta^{\pm}_0\omega_k^{23}-\Delta^{\pm}_2\omega_k^{03}+\Delta^{\pm}_3\omega_k^{02})e_{023}^k\\+(\Delta^{\pm}_1\omega_k^{23}-\Delta^{\pm}_2\omega_k^{13}+\Delta^{\pm}_3\omega_k^{12})e_{123}^k].
\end{align*}
From this using \eqref{2.12} and \eqref{2.28} we have
\begin{equation*}
d^c\overset{2}{\omega}=\Big(\sum_{\mu=0}^3e_\mu\Delta^+_\mu\overset{2}{\omega}\Big)_{3}
\end{equation*}
and
\begin{equation*}
\delta^c\overset{2}{\omega}=\Big(\sum_{\mu=0}^3e_\mu\Delta^-_\mu\overset{2}{\omega}\Big)_{1}.
\end{equation*}
For $\overset{4}{\omega}$ we obtain
\begin{align*}
\sum_{\mu=0}^3e_\mu\Delta^-_\mu\overset{4}{\omega}=\sum_k(\Delta^-_0\overset{4}{\omega}_ke_{123}^k+\Delta^-_1\overset{4}{\omega}_ke_{023}^k-\Delta^-_2\overset{4}{\omega}_ke_{013}^k+
\Delta^-_3\overset{4}{\omega}_ke_{012}^k).
\end{align*}
By \eqref{2.30}, we have
\begin{equation*}
\delta^c\overset{4}{\omega}=\sum_{\mu=0}^3e_\mu\Delta^-_\mu\overset{4}{\omega}=\Big(\sum_{\mu=0}^3e_\mu\Delta^-_\mu\overset{4}{\omega}\Big)_{3}.
\end{equation*}
Similar calculations apply to the case of odd forms. Hence
\begin{equation}\label{3.7}
d^c\overset{r}{\omega}=\Big(\sum_{\mu=0}^3e_\mu\Delta^+_\mu\overset{r}{\omega}\Big)_{r+1} \quad  \mbox{for} \quad  r=0,1,2,3
\end{equation}
and
\begin{equation}\label{3.8}
\delta^c\overset{r}{\omega}=\Big(\sum_{\mu=0}^3e_\mu\Delta^-_\mu\overset{r}{\omega}\Big)_{r-1} \quad  \mbox{for} \quad  r=1,2,3,4.
\end{equation}
Combining \eqref{3.7} and \eqref{3.8} yields \eqref{3.6}.
\end{proof}

Clearly,  the discrete Dirac-K\"{a}hler equation \eqref{3.2}  can be rewritten in the form
\begin{equation*}
i\Big(\sum_{r=0}^3\Big(\sum_{\mu=0}^3e_\mu\Delta^+_\mu\overset{r}{\omega}\Big)_{r+1}+
\sum_{r=1}^4\Big(\sum_{\mu=0}^3e_\mu\Delta^-_\mu\overset{r}{\omega}\Big)_{r-1}\Big)=m\sum_{r=0}^4\overset{r}{\omega}.
\end{equation*}
Let $K^{ev}(4)=K^0(4)\oplus K^2(4)\oplus K^4(4)$ and let $\Omega^{ev}\in K^{ev}(4)$ be a complex-valued even inhomogeneous form, i.e. $\Omega^{ev}=\overset{0}{\omega}+\overset{2}{\omega}+\overset{4}{\omega}$.
A discrete analogue of the Joyce equation \eqref{1.8} is defined by
\begin{equation}\label{3.9}
i(d^c+\delta^c)\Omega^{ev}=m\Omega^{ev}e_0,
\end{equation}
where $e_0$ is given by  \eqref{3.4}. From \eqref{3.6} it follows that Eq.~\eqref{3.9} is equivalent to
\begin{equation*}
i\Big(\sum_{\mu=0}^3e_\mu\Delta^+_\mu\overset{0}{\omega}+\Big(\sum_{\mu=0}^3e_\mu\Delta^+_\mu\overset{2}{\omega}\Big)_3+
\Big(\sum_{\mu=0}^3e_\mu\Delta^-_\mu\overset{2}{\omega}\Big)_1+\sum_{\mu=0}^3e_\mu\Delta^-_\mu\overset{4}{\omega}\Big)=m\Omega^{ev}e_0.
\end{equation*}
A more detailed calculation leads to the following system of 8 difference equations
\begin{align*}\label{}
i(\Delta^+_0\overset{0}{\omega}_k+\Delta^-_1\omega_k^{01}+\Delta^-_2\omega_k^{02}+\Delta^-_3\omega_k^{03})=m\overset{0}\omega_k,\\
i(\Delta^+_1\overset{0}{\omega}_k+\Delta^-_0\omega_k^{01}+\Delta^-_2\omega_k^{12}+\Delta^-_3\omega_k^{13})=-m\omega_k^{01},\\
i(\Delta^+_2\overset{0}{\omega}_k+\Delta^-_0\omega_k^{02}-\Delta^-_1\omega_k^{12}+\Delta^-_3\omega_k^{23})=-m\omega_k^{02},\\
i(\Delta^+_3\overset{0}{\omega}_k+\Delta^-_0\omega_k^{03}-\Delta^-_1\omega_k^{13}-\Delta^-_2\omega_k^{23})=-m\omega_k^{03},\\
i(\Delta^+_0\omega_k^{12}-\Delta^+_1\omega_k^{02}+\Delta^+_2\omega_k^{01}+\Delta^-_3\overset{4}{\omega}_k)=m\omega_k^{12},\\
i(\Delta^+_0\omega_k^{13}-\Delta^+_1\omega_k^{03}+\Delta^+_3\omega_k^{01}-\Delta^-_2\overset{4}{\omega}_k)=m\omega_k^{13},\\
i(\Delta^+_0\omega_k^{23}-\Delta^+_2\omega_k^{03}+\Delta^+_3\omega_k^{02}+\Delta^-_1\overset{4}{\omega}_k)=m\omega_k^{23},\\
i(\Delta^+_1\omega_k^{23}-\Delta^+_2\omega_k^{13}+\Delta^+_3\omega_k^{12}+\Delta^-_0\overset{4}{\omega}_k)=-m\overset{4}{\omega}_k
\end{align*}
for each $k=(k_0, k_1, k_2, k_3)$.

\section{Plane Wave Solutions}
In this section, we consider solutions of the discrete  Joyce equation which imitate the plane wave solutions for the  continuum counterpart.
We will mainly follow the strategy of \cite{S4}, but there are extra difficulties here. These difficulties arise in the construction of eigenvectors of the operator  $i(d^c+\delta^c)$ since this operator consists of  difference operators of the two types.

Let us consider the complex-valued 0-forms
\begin{equation}\label{4.1}
 \psi^0=\sum_k\psi^0_kx^k,\qquad \psi^{\mu\nu}=\sum_k\psi^{\mu\nu}_kx^k, \qquad \psi^4=\sum_k\psi^4_kx^k,
 \end{equation}
 where
 \begin{equation*}\label{}
  \psi^0_k=(1+ip_0)^{k_0}(1+ip_1)^{k_1}(1+ip_2)^{k_2}(1+ip_3)^{k_3},
 \end{equation*}
 \begin{equation*}\label{}
  \psi^{01}_k=(1-ip_0)^{-k_0}(1-ip_1)^{-k_1}(1+ip_2)^{k_2}(1+ip_3)^{k_3},
 \end{equation*}
  \begin{equation*}\label{}
  \psi^{02}_k=(1-ip_0)^{-k_0}(1+ip_1)^{k_1}(1-ip_2)^{-k_2}(1+ip_3)^{k_3},
 \end{equation*}
  \begin{equation*}\label{}
  \psi^{03}_k=(1-ip_0)^{-k_0}(1+ip_1)^{k_1}(1+ip_2)^{k_2}(1-ip_3)^{-k_3},
 \end{equation*}
  \begin{equation*}\label{}
  \psi^{12}_k=(1+ip_0)^{k_0}(1-ip_1)^{-k_1}(1-ip_2)^{-k_2}(1+ip_3)^{k_3},
 \end{equation*}
  \begin{equation*}\label{}
  \psi^{13}_k=(1+ip_0)^{k_0}(1-ip_1)^{-k_1}(1+ip_2)^{k_2}(1-ip_3)^{-k_3},
 \end{equation*}
  \begin{equation*}\label{}
  \psi^{23}_k=(1+ip_0)^{k_0}(1+ip_1)^{k_1}(1-ip_2)^{-k_2}(1-ip_3)^{-k_3},
 \end{equation*}
  \begin{equation*}\label{}
  \psi^4_k=(1-ip_0)^{-k_0}(1-ip_1)^{-k_1}(1-ip_2)^{-k_2}(1-ip_3)^{-k_3},
 \end{equation*}
 and $p_\mu\in\mathbb{R}$.
 It is easy to check that
 \begin{equation*}
  \Delta^+_{\mu}\big((1+ip_{\mu})^{k_{\mu}}\big)=(1+ip_{\mu})^{k_{\mu}+1}-(1+ip_{\mu})^{k_{\mu}}=ip_{\mu}(1+ip_{\mu})^{k_{\mu}},
 \end{equation*}
 \begin{equation*}
  \Delta^-_{\mu}\big((1-ip_{\mu})^{-k_{\mu}}\big)=(1-ip_{\mu})^{-k_{\mu}}-(1-ip_{\mu})^{-(k_{\mu}-1)}=ip_{\mu}(1-ip_{\mu})^{-k_{\mu}}.
 \end{equation*}
As a consequence we obtain
 \begin{equation}\label{4.2}
  \Delta^+_{\mu}\psi^0_k=ip_{\mu}\psi^0_k, \qquad \Delta^-_{\mu}\psi^4_k=ip_{\mu}\psi^4_k,
 \end{equation}
 \begin{equation}\label{4.3}
  \Delta^-_{\mu}\psi^{\mu\nu}_k=ip_{\mu}\psi^{\mu\nu}_k, \qquad \Delta^-_{\nu}\psi^{\mu\nu}_k=ip_{\nu}\psi^{\mu\nu}_k,
 \end{equation}
 \begin{equation}\label{4.4}
  \Delta^+_{\iota}\psi^{\mu\nu}_k=ip_{\iota}\psi^{\mu\nu}_k, \quad \mbox{for} \quad \iota\neq\mu, \nu.
 \end{equation}

Let  $A\in K^{ev}(4)$ is a constant complex-valued form.
Hence $A$ can be expanded as
 \begin{equation}\label{4.5}
  A=\alpha^0x+\sum_{\mu<\nu}\alpha^{\mu\nu}e_{\mu\nu}+\alpha^4e,
 \end{equation}
   where $\alpha^0, \alpha^{\mu\nu}, \alpha^4\in\mathbb{C}$  and $x$, $e_{\mu\nu}$, $e$ are the unit forms  given by  \eqref{3.4}.
   Consider the form
   \begin{equation}\label{4.6}
  \Phi=\overset{0}\varphi+\overset{2}\varphi+\overset{4}\varphi,
 \end{equation}
  where
 \begin{equation*}
\overset{0}\varphi=\alpha^0\psi^0, \quad \overset{2}\varphi=\sum_{\mu<\nu}\alpha^{\mu\nu}\psi^{\mu\nu}e_{\mu\nu}, \quad \overset{4}\varphi=\alpha^4\psi^4e,
\end{equation*}
and $\psi^0$, $\psi^{\mu\nu}$ and $\psi^{4}$ are given by \eqref{4.1}.

Now we apply the operators $d^c$ and $\delta^c$ to the forms $\overset{0}\varphi$, $\overset{2}\varphi$ and $\overset{4}\varphi$.
Using \eqref{3.7}, \eqref{3.8} and \eqref{4.2} we obtain
 \begin{equation*}
d^c\overset{0}\varphi=\sum_{\mu=0}^3e_\mu\Delta^+_\mu\overset{0}\varphi=\sum_{\mu=0}^3e_\mu\alpha^0\sum_k(\Delta^+_\mu\psi^0_k)x^k=i\Big(\sum_{\mu=0}^3e_\mu p_\mu\Big)\overset{0}\varphi
\end{equation*}
and
\begin{equation*}
\delta^c\overset{4}\varphi=\sum_{\mu=0}^3e_\mu\Delta^-_\mu\overset{4}\varphi=\sum_{\mu=0}^3e_\mu\alpha^4\sum_k(\Delta^-_\mu\psi^4_k)x^ke=
i\Big(\sum_{\mu=0}^3e_\mu p_\mu\Big)\overset{4}\varphi.
\end{equation*}
By \eqref{4.3} and \eqref{4.4}, as in the proof Proposition~3.2 we calculate
\begin{align*}
d^c\overset{2}{\varphi}=\Big(\sum_{\mu=0}^3e_\mu\Delta^+_\mu\overset{2}{\varphi}\Big)_{3}=
ip_0e_0(\alpha^{12}\psi^{12}e_{12}+\alpha^{13}\psi^{13}e_{13}+\alpha^{23}\psi^{23}e_{23})\\+
ip_1e_1(\alpha^{02}\psi^{02}e_{02}+\alpha^{03}\psi^{03}e_{03}+\alpha^{23}\psi^{23}e_{23})\\+
ip_2e_2(\alpha^{01}\psi^{01}e_{01}+\alpha^{03}\psi^{03}e_{03}+\alpha^{13}\psi^{13}e_{13})\\+
ip_3e_3(\alpha^{01}\psi^{01}e_{01}+\alpha^{02}\psi^{02}e_{02}+\alpha^{12}\psi^{12}e_{12})
\end{align*}
and
\begin{align*}
\delta^c\overset{2}{\varphi}=\Big(\sum_{\mu=0}^3e_\mu\Delta^-_\mu\overset{2}{\varphi}\Big)_{1}=
ip_0e_0(\alpha^{01}\psi^{01}e_{01}+\alpha^{02}\psi^{02}e_{02}+\alpha^{03}\psi^{03}e_{03})\\+
ip_1e_1(\alpha^{01}\psi^{01}e_{01}+\alpha^{12}\psi^{12}e_{12}+\alpha^{13}\psi^{13}e_{13})\\+
ip_2e_2(\alpha^{12}\psi^{12}e_{12}+\alpha^{02}\psi^{02}e_{02}+\alpha^{23}\psi^{23}e_{23})\\+
ip_3e_3(\alpha^{03}\psi^{03}e_{03}+\alpha^{13}\psi^{13}e_{13}+\alpha^{23}\psi^{23}e_{23}).
\end{align*}
This yields
\begin{equation*}
(d^c+\delta^c)\overset{2}\varphi=i\Big(\sum_{\mu=0}^3e_\mu p_\mu\Big)\overset{2}\varphi.
\end{equation*}
Therefore, we have the following result.
\begin{prop}
  For the form \eqref{4.6} we have that
  \begin{equation}\label{4.7}
(d^c+\delta^c)\Phi=i\Big(\sum_{\mu=0}^3e_\mu p_\mu\Big)\Phi.
\end{equation}
  \end{prop}
Substituting \eqref{4.7} into Eq.~\eqref{3.9}  we obtain
 \begin{equation}\label{4.8}
   -\Big(\sum_{\mu=0}^3e_\mu p_\mu\Big) \Phi=m\Phi e_0.
 \end{equation}
 It is not difficult to show that Eq.~\eqref{4.8} is equivalent to the following system of equations
 \begin{align}
(p_0+m)\alpha^{0}\psi^{0}+p_1\alpha^{01}\psi^{01}+p_2\alpha^{02}\psi^{02}+p_3\alpha^{03}\psi^{03}=0,\label{4.9}\\
(p_0+m)\alpha^{12}\psi^{12}-p_1\alpha^{02}\psi^{02}+p_2\alpha^{01}\psi^{01}-p_3\alpha^{4}\psi^{4}=0,\label{4.10}\\
(p_0+m)\alpha^{13}\psi^{13}-p_1\alpha^{03}\psi^{03}-p_2\alpha^{4}\psi^{4}+p_3\alpha^{01}\psi^{01}=0,\label{4.11}\\
(p_0+m)\alpha^{23}\psi^{02}+p_1\alpha^{4}\psi^{4}-p_2\alpha^{03}\psi^{03}+p_3\alpha^{02}\psi^{02}=0,\label{4.12}\\
(p_0-m)\alpha^{01}\psi^{01}+p_1\alpha^{0}\psi^{0}+p_2\alpha^{12}\psi^{12}+p_3\alpha^{13}\psi^{13}=0,\label{4.13}\\
(p_0-m)\alpha^{02}\psi^{02}-p_1\alpha^{12}\psi^{12}+p_2\alpha^{0}\psi^{0}+p_3\alpha^{23}\psi^{23}=0,\label{4.14}\\
(p_0-m)\alpha^{03}\psi^{03}-p_1\alpha^{13}\psi^{13}-p_2\alpha^{23}\psi^{23}+p_3\alpha^{0}\psi^{0}=0,\label{4.15}\\
(p_0-m)\alpha^{4}\psi^{4}+p_1\alpha^{23}\psi^{23}-p_2\alpha^{13}\psi^{13}+p_3\alpha^{12}\psi^{12}=0.\label{4.16}
\end{align}
From \eqref{3.5} it follows that $e_0e_0=x$. Hence Eq.~\eqref{4.8} can be written as
 \begin{equation}\label{4.17}
   -\Big(p_0x+\sum_{\mu=1}^3p_\mu e_0 e_\mu\Big)\Phi=me_0\Phi e_0.
 \end{equation}
 A direct computation shows that
 \begin{equation*}
   \Big(p_0x-\sum_{\mu=1}^3p_\mu e_0 e_\mu\Big)\Big(p_0x+\sum_{\mu=1}^3p_\mu e_0 e_\mu\Big)=\Big(p_0^2-\sum_{\mu=1}^3p_\mu^2\Big)x.
 \end{equation*}
Then
 multiplying both sides of \eqref{4.17} by the same factor: $-\big(p_0x-\sum_{\mu=1}^3p_\mu e_0 e_\mu\big)$ gives
 \begin{equation*}
  \Big(p_0^2-\sum_{\mu=1}^3p_\mu^2\Big)x\Phi=-m\Big(p_0x-\sum_{\mu=1}^3p_\mu e_0 e_\mu\Big)e_0\Phi e_0.
 \end{equation*}
 By \eqref{3.5} this gives the equation
  \begin{equation*}
  \Big(p_0^2-\sum_{\mu=1}^3p_\mu^2\Big)\Phi=-m\Big(\sum_{\mu=0}^3p_\mu e_\mu\Big)\Phi e_0.
 \end{equation*}
 Applying \eqref{4.8} to the right-hand side   we obtain
 \begin{equation*}
  \Big(p_0^2-\sum_{\mu=1}^3p_\mu^2\Big)\Phi=m^2\Phi e_0e_0,
 \end{equation*}
 or equivalently,
 \begin{equation*}
  \Big(p_0^2-\sum_{\mu=1}^3p_\mu^2-m^2\Big)\Phi=0.
 \end{equation*}
  Thus exactly as in \cite{S4}, we have the following assertion.
  \begin{prop}
  The form \eqref{4.6} is a non-trivial solution of Eq.~\eqref{3.9} if and only if
  \begin{equation}\label{4.18}
  p_0^2=m^2+p_1^2+p_2^2+p_3^2.
 \end{equation}
   \end{prop}
 Therefore, if we define  $p=\{ p_0, p_1, p_2, p_3\}$ to be the energy-momentum vector of a particle with (proper) mass $m$,  then the relation \eqref{4.18} is the energy-momentum relation.
 In view of this, it makes sense to state that  the form \eqref{4.6} is a discrete version of the plane wave solution.
 Recall that  the Joyce equation \eqref{1.6} admits the plane wave solutions of the form \eqref{1.7}. Thus the 0-form $\psi=\psi^{0}+\psi^{\mu\nu}+\psi^{4}$ given by \eqref{4.1} plays a role of the function $e^{ip\cdot x}$ in the discrete case.

It should be noted that the same proposition for a discrete analog of the Hestenes equation in the case of a discrete model based on the double complex construction is proven in \cite{S3}.

   Let us represent the even complex-valued form  \eqref{4.5} as
\begin{equation*}
  A=A_{+}+A_{-},
 \end{equation*}
 where
 \begin{equation*}
   A_{+}=\alpha^{0}x+\alpha^{12}e_{12}+\alpha^{13}e_{13}+\alpha^{23}e_{23},
  \end{equation*}
 \begin{equation*}
   A_{-}=\alpha^{01}e_{01}+\alpha^{02}e_{02}+\alpha^{03}e_{03}+\alpha^4e.
  \end{equation*}
  It is easy to check that $A_{+}$ commutes with $e_0$  and  $A_{-}$  anticommutes with it, i.e.
  \begin{equation}\label{4.19}
  e_0A_{\pm}=\pm A_{\pm}e_0.
 \end{equation}
Then the form  \eqref{4.6} can be represent also as
  \begin{equation*}
  \Phi=\Phi_{+}+\Phi_{-},
  \end{equation*}
  where
 \begin{equation}\label{4.20}
   \Phi_{+}=\alpha^{0}\psi^{0}x+\alpha^{12}\psi^{12}e_{12}+\alpha^{13}\psi^{13}e_{13}+\alpha^{23}\psi^{23}e_{23},
  \end{equation}
 \begin{equation}\label{4.21}
  \Phi_{-}=\alpha^{01}\psi^{01}e_{01}+\alpha^{02}\psi^{02}e_{02}+\alpha^{03}\psi^{03}e_{03}+\alpha^4\psi^{4}e.
  \end{equation}
  Since the 0-forms \eqref{4.1} commute with $e_0$  then  Eq.~\eqref{4.19} is true also for $\Phi_{\pm}$, i.e.
 \begin{equation*}
  e_0\Phi_{\pm}=\pm \Phi_{\pm}e_0.
 \end{equation*}
 As is shown in \cite{S4}, the proof of the following lemma follows from a direct computation.
 \begin{lem}
  The form $e_{0\mu}\Phi_{-}$ commutes with $e_0$ and has the view    \eqref{4.20},  while
  $e_{0\mu}\Phi_{+}$ anticommutes with $e_0$ and has the view    \eqref{4.21}  for any $\mu=1,2,3$.
  \end{lem}

  By Lemma~4.3, we have the following result whose proof may be found in \cite{S4}.
  \begin{thm}
 The form $\Phi$ given by \eqref{4.6} is a   non-trivial solution of the discrete Joyce equation if and only if the condition
\begin{equation}\label{4.22}
 \Phi_{-}=\frac{p_1e_{01}+p_2e_{02}+p_3e_{03}}{m-p_0}\Phi_{+}
 \end{equation}
 holds,  or equivalently,
 \begin{equation}\label{4.23}
 \Phi_{+}=-\frac{p_1e_{01}+p_2e_{02}+p_3e_{03}}{m+p_0}\Phi_{-}.
 \end{equation}
\end{thm}

It is not difficult to show that  the condition \eqref{4.22} is equivalent to the system of Eqs.~\eqref{4.13}--\eqref{4.16}. Similarly,  the condition \eqref{4.23} gives the system of Eqs.~\eqref{4.9}--\eqref{4.12}.
 Moreover,
 using Eqs.~\eqref{4.13}--\eqref{4.16} the form \eqref{4.22} can be written as
 \begin{eqnarray*}
 \Phi_{-}=\frac{p_1\alpha^{0}\psi^{0}+p_2\alpha^{12}\psi^{12}+p_3\alpha^{13}\psi^{13}}{m-p_0}e_{01}
 +\frac{p_2\alpha^{0}\psi^{0}-p_1\alpha^{12}\psi^{12}+p_3\alpha^{23}\psi^{23}}{m-p_0}e_{02}\\
 +\frac{p_3\alpha^{0}\psi^{0}-p_1\alpha^{13}\psi^{13}-p_2\alpha^{23}\psi^{23}}{m-p_0}e_{03}+\frac{p_3\alpha^{12}\psi^{12}-p_2\alpha^{13}\psi^{13}+p_1\alpha^{23}\psi^{23}}{m-p_0}e.
   \end{eqnarray*}
   Hence
   \begin{eqnarray}\label{4.24}
\Phi=a_1\psi^{0}((m-p_0)x+p_1e_{01}+p_2e_{02}+p_3e_{03})\nonumber\\+a_2\psi^{12}((m-p_0)e_{12}+p_2e_{01}-p_1e_{02}+p_3e)\nonumber \\
  +a_3\psi^{13}((m-p_0)e_{13}+p_3e_{01}-p_1e_{03}+p_2e)\nonumber\\
    +a_4\psi^{23}((m-p_0)e_{23}+p_3e_{02}-p_2e_{03}+p_1e),
 \end{eqnarray}
 where
 \begin{equation*}
   a_1=\frac{\alpha^{0}}{m-p_0}, \quad a_2=\frac{\alpha^{12}}{m-p_0}, \quad a_3=\frac{\alpha^{13}}{m-p_0}, \quad a_4=\frac{\alpha^{23}}{m-p_0}. \quad
 \end{equation*}
It turns out that  the discrete general plane wave solution of Eq.~\eqref{3.9} consists of four linearly independent solutions for given $p_\mu$, $\mu=1,2,3$.
To be more precise, there are four linearly independent solutions  for each positive and  negative  $p_0=\pm\sqrt{m^2+p_1^2+p_2^2+p_3^2}$.
Thus the result is the same as for the continuum counterpart. The presence of eight such solutions in the continuum case is interpreted in details  in \cite{J}.

Similarly, again, if we take the condition \eqref{4.23} and use Eqs.~\eqref{4.9}--\eqref{4.12} then we obtain
\begin{eqnarray*}
 \Phi_{+}=\frac{-p_1\alpha^{01}\psi^{01}-p_2\alpha^{02}\psi^{02}-p_3\alpha^{03}\psi^{03}}{m+p_0}x
 +\frac{p_1\alpha^{02}\psi^{02}-p_2\alpha^{01}\psi^{01}-p_3\alpha^{4}\psi^{4}}{m+p_0}e_{12}\\
 +\frac{p_1\alpha^{03}\psi^{03}+p_2\alpha^{4}\psi^{4}-p_3\alpha^{01}\psi^{01}}{m+p_0}e_{13}+\frac{-p_1\alpha^{4}\psi^{4}+p_2\alpha^{03}\psi^{03}-p_3\alpha^{02}\psi^{02}}{m+p_0}e_{23}.
   \end{eqnarray*}
This leads to the general solution
\begin{eqnarray}\label{4.25}
\Phi=b_1\psi^{01}((m+p_0)e_{01}-p_1x-p_2e_{12}-p_3e_{13})\nonumber\\\
+b_2\psi^{02}((m+p_0)e_{02}-p_2x+p_1e_{12}-p_3e_{23})\nonumber\ \\
  +b_3\psi^{03}((m+p_0)e_{03}-p_3x+p_1e_{13}+p_2e_{23})\nonumber\\\
    +b_4\psi^{4}((m+p_0)e-p_3e_{12}+p_2e_{13}-p_1e_{23}),
 \end{eqnarray}
 where
 \begin{equation*}
   b_1=\frac{\alpha^{01}}{m+p_0}, \quad b_2=\frac{\alpha^{02}}{m+p_0}, \quad b_3=\frac{\alpha^{03}}{m+p_0}, \quad b_4=\frac{\alpha^{4}}{m+p_0}. \quad
 \end{equation*}
 It is clear that the equivalence of \eqref{4.22} and \eqref{4.23} implies the equivalence of the solutions \eqref{4.24} and \eqref{4.25}.

\medskip


\end{document}